\documentclass[conf,twocolumn]{IEEEtran}
\usepackage{graphicx,cite,calc,dsfont}
\usepackage[defs]{ams}
\usepackage{eqnarray,subeqn,xcolor}
\usepackage{hyperref}
\usepackage{amsthm,amssymb}
\usepackage{algorithm}
\usepackage{algpseudocode}
\newtheorem{thm}{Theorem}
\newtheorem{lem}{Lemma}

\theoremstyle{definition}

\newtheorem{myremark}{Remark}

\begin{document}

\title{Caching Gain in Wireless Networks with Fading:\\ A Multi-User Diversity Perspective}

\author{ Seyed Pooya Shariatpanahi$^1$ , Hamed Shah-Mansouri, Babak Hossein Khalaj$^2$ \\
1: School of Computer Science, Institute for Research in Fundamental Sciences (IPM), Tehran, Iran. \\
2: Department of Electrical Engineering and Advanced Communication Research Institute (ACRI)\\
Sharif University of Technology, Tehran, Iran.\\
(emails: pooya@ipm.ir, hshahmansour@alum.sharif.edu, khalaj@sharif.edu)\\ }

\maketitle
\thispagestyle{empty}
\pagestyle{empty}

\section*{Abstract}
\emph{We consider the effect of caching in wireless networks where fading is the dominant channel effect.
First, we propose a one-hop transmission strategy for cache-enabled wireless networks, which is based on exploiting multi-user diversity gain.
Then, we derive a closed-form result for throughput scaling of the proposed scheme in large networks, which reveals the inherent trade-off between cache memory size and network throughput.
Our results show that substantial throughput improvements are achievable in networks with sources equipped with large cache size.
We also verify our analytical result through simulations.
\let\thefootnote\relax\footnotetext{This research was in part supported by a grant from IPM.}}\\ \\
\emph{\textbf{Index Terms}}--- Caching, fading, order statistics, multi-user diversity, throughput scaling, wireless networks.

\section{Introduction}\label{Sec_Introduction}

Almost all wireless communication networks face two different conditions of operation, namely, off-peak hours and peak hours.
During the off-peak hours, the users do not request much traffic, and accordingly, the network is not crowded.
In contrast, during peak hours, the network is crowded and network capacity limitation becomes of high importance.
A promising idea to overcome the difficulties encountered in network operation during peak hours is to cache data at end-user devices memory during off-peak hours, and exploit these cached data at peak hours to relieve network congestion (\cite{Maddah_2012}, \cite{Maddah_2013}, \cite{Niesen_2013}, and \cite{Azimdoost_2012}).

While this idea looks very attractive, in order to devise efficient network operation protocols in this context, there is the need for understanding the costs and benefits this idea brings to the network.
In this regard, recently, a number of papers have analyzed applying this idea to overcome the capacity limitation in wireless networks.
In \cite{Azimdoost_2012}, throughput of cache-enabled large networks for grid and random node placement is investigated, where the cached data has a limited life-time.
In \cite{Maddah_2012}, \cite{Maddah_2013}, and \cite{Niesen_2013} a new information-theoretic framework for the caching problem is developed.
Also, in \cite{Gitzenis_2013}, the performance of multi-hopping on a square grid, and the effect of caching, on network performance, is analyzed.
\cite{Golrezai_2012} and \cite{Ji_Basic_Priciples} focus on device-to-device caching networks where, for example, in \cite{Ji_Basic_Priciples}, the effect of coded multi-casting gain versus spatial reuse gain is investigated.
Finally, \cite{Ji_One_Hop} considers a single-hop caching network, and considering an interference avoidance strategy, it investigates the optimal throughput-outage trade-off.

In this paper, we propose another perspective on analyzing caching gain in wireless networks, namely, in the context of multi-user diversity.
Consider a wireless network with $n$ sources and $n$ destinations, communicating in a shared wireless medium.
Suppose there exists a pool of data files cached at transmitter nodes, and each receiver is interested in accessing one of such data files.
Also, suppose that the wireless channel power gain between sources and destinations are modelled by i.i.d. random variables -- the so-called \emph{random connections model} (\cite{Gowikar}, \cite{Cui}, \cite{Pooya_WCL}, \cite{Pooya_2013}, and \cite{Cui_Opportunistic}).
In such model, random fading of wireless channel is the dominant effect, in contrast to the models where the path loss effect is dominant.
Thus, we face a critically interference-limited network, and our objective is to deliver as many files as we can to the destinations.

In previous works, two different viewpoints on caching gain are considered.
In \cite{Maddah_2012}, \cite{Maddah_2013}, and \cite{Niesen_2013} the caching gain is achieved by exploiting the coded multi-casting gain, while in \cite{Azimdoost_2012}, \cite{Gitzenis_2013}, \cite{Golrezai_2012}, \cite{Ji_Basic_Priciples}, and \cite{Ji_One_Hop}, caching gain is exploited via the spatial reuse gain.
In contrast, due to fading dominance in our network channel model, we encounter a critically interference-limited network.
Accordingly, we propose to analyze the caching effect in the context of multi-user diversity gain.
This gain is the result of an opportunistic transmission strategy, benefiting from random fluctuations of the wireless channel.

We consider one-hop communications in a cached-enabled wireless network under the random connections model.
In our model, the sources follow an on-off transmission paradigm, and the receivers use single-user decoding techniques.
Our proposed transmission strategy (presented in the form of an algorithm) is to find out the maximum number of possible transmissions with best channel conditions, in order to maximize the number of successful receptions.
We derive a closed-form scaling expression for the average throughput of a network with heavy-tailed channel power distribution.
With the help of such expression, we explain the trade-off between memory and throughput in a cache-enabled wireless network, where fading is the dominant channel effect. Also, we provide simulations verifying our analytical results.

The rest of the paper is organized as follows.
In Section \ref{Sec_Model}, we describe the network model.
In Section \ref{Sec_Strategy}, we present the main algorithm according to which the network transmission strategy is determined.
We analyze network performance in terms of throughput, in Section \ref{Sec_Throughput}.
In Section \ref{Sec_Simulations}, we present simulations and finally,  Section \ref{Summary} concludes the paper.

\section{Network Model}\label{Sec_Model}
Consider a wireless network consisting of $n$ source (provider) and $n$ destination (client) nodes. We denote the source nodes by $S_1, \dots, S_n$ and the destination nodes by $D_1, \dots, D_n$. Consider a pool of $n$ data files named $F_1, \dots, F_n$. We assume that the destination $D_i$ is interested in accessing the data file $F_i$ from that pool, via one-hop transmission from the sources. Each source $S_i$ has a library of $m$ files from the pool, cached in its memory. Thus, each source has the potential to serve one of $m$ destinations who request the file available in the source's cache. We call the set of destinations that can be served by the source $S_i$ as $\Delta_i$ (see Fig. \ref{Fig1} for an example of the case $n=4$ and $m=2$). In this paper, we assume that the source caches are filled with the files from the pool randomly in a uniform way, without considering destination requests.

The wireless channel power between source $S_i$ and destination $D_j$ is characterized by the random variable $\gamma_{i,j}$. Therefore, the network channel state is completely characterized by the $n \times n$ channel power matrix $\mathbf{\Gamma}$ whose elements are $\gamma_{i,j}$. Also, we assume that the elements of $\mathbf{\Gamma}$ are independently and identically distributed (i.i.d.) from the distribution $f(\gamma)$ (similar to \cite{Gowikar}, \cite{Cui}, \cite{Pooya_WCL}, \cite{Pooya_2013}, and \cite{Cui_Opportunistic}), and follow a quasi-static rule. In other words, the elements of $\mathbf{\Gamma}$ remain fixed during each time slot, and change in the subsequent time slots, independent from other time slots.


The sources are assumed to follow an on-off transmission paradigm. At each time slot, a subset of sources will be activated (with unit power), and each of them will send one of the data files from its library, while other sources remain silent in that time slot. We denote the set of activated sources by $\mathbb{S}$. The set of sources which are active at each time slot (i.e. $\mathbb{S}$) and the target destination of each source constitute the \emph{transmission strategy} of that time slot. \emph{Transmission strategy} at each time slot should be determined based on the Channel State Information (CSI) encompassed in channel power gain matrix. In this paper, we assume that a centralized supervisor, which has access to full CSI, determines the \emph{transmission strategy}. Making this decision in a distributed manner is out of the scope of this paper.

We assume one-hop transmission conducted during a single time slot, and assume that destinations use only single-user decoding techniques. Thus, in order to have a successful transmission from source $S_i$ to the destination $D_j$, the receive Signal to Interference and Noise Ratio ($SINR$) at $D_j$ should be above a constant threshold level $\beta$. Consequently, the successful transmission condition from $S_i$ to $D_j$ can be written as:
\begin{equation}\label{Eq_Model_SINR_Constraint}
SINR_{i,j} \triangleq \frac{\gamma_{i,j}}{N_0+\sum_{S_k \in \mathbb{S}, k \neq i}{\gamma_{k,j}}} \geq \beta.
\end{equation}
\begin{figure}
\begin{center}
\includegraphics[width=0.45\textwidth]{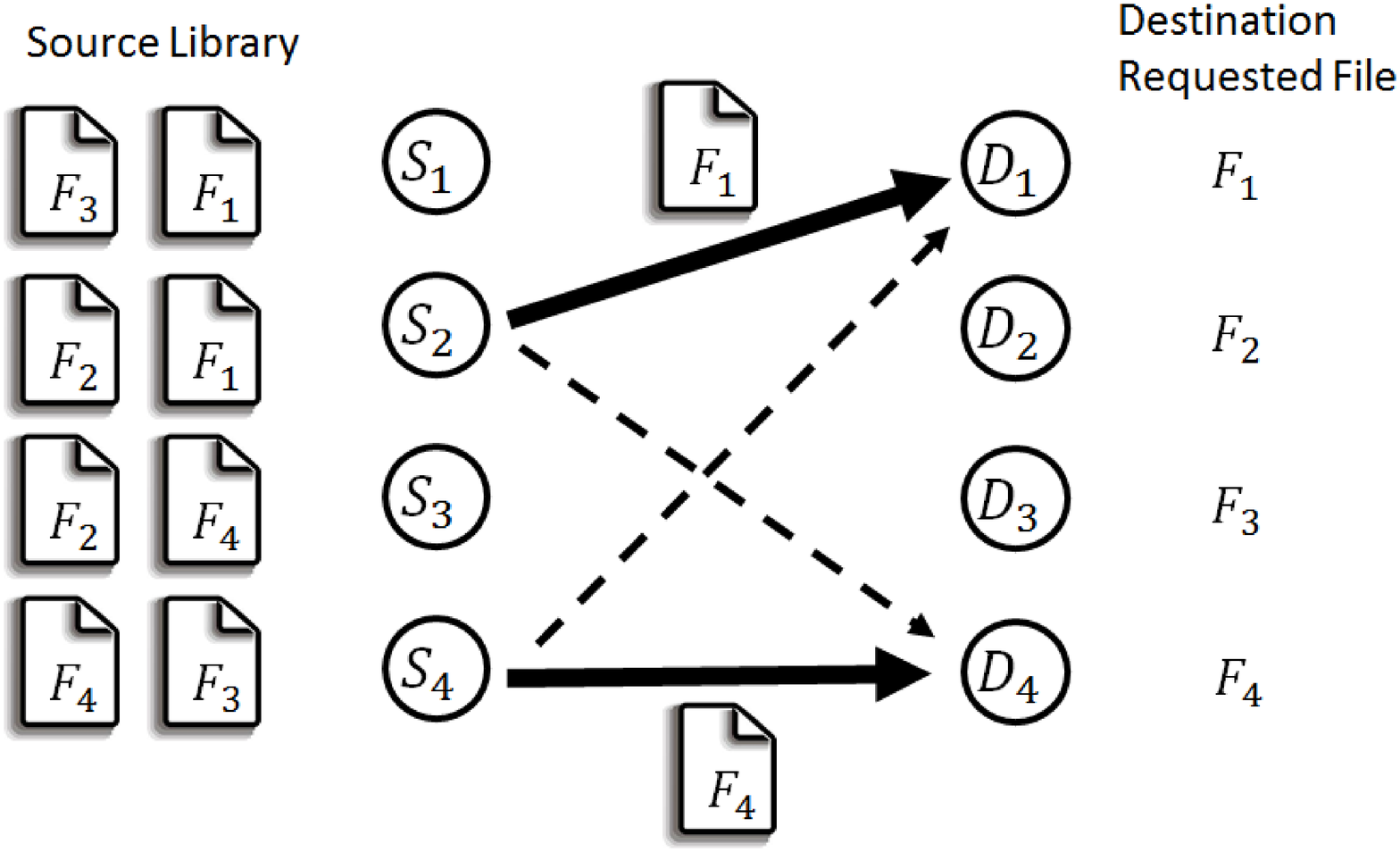}
\end{center}
\caption{Each source has a library of files. Each destination requests one file. In this simple example, we have $n=4$, $m=2$, $\Delta_1=\{D_1, D_3\}$, $\Delta_2=\{D_1, D_2\}$, $\Delta_3=\{D_2, D_4\}$, and $\Delta_4=\{D_3, D_4\}$. The activation strategy is set to $\mathbb{S}=\{S_2, S_4\}$. Solid arrows show direct links, while, dashed arrows show interference links. \label{Fig1}}
\end{figure}
At each time slot, network throughput $T$ is defined as the number of successful transmissions from the sources to the destinations. Since channel power gains are modelled by random variables, network throughput $T$ is also a random variable.

In this paper, we assume that the channel power has the following p.d.f. and c.d.f., respectively (similar to \cite{Gowikar} \cite{Cui}, and \cite{Pooya_WCL})
\begin{eqnarray}\label{Eq_Model_Pareto_PDF_CDF}
f(x) &=& \left\{ \begin{array}{ll}
         0 & \mbox{, $x < 1$},\\
        \alpha/x^{\alpha+1} & \mbox{, $x \geq 1$}.\end{array} \right. \\ \nonumber
F(x) &=& \left\{ \begin{array}{ll}
         0 & \mbox{, $x < 1$},\\
        1-1/x^\alpha & \mbox{, $x \geq 1$}.\end{array} \right. \\ \nonumber
\end{eqnarray}
where $\alpha \geq 2$ exhibits heavy-tailed behaviour. In this paper, we have considered the distribution in (\ref{Eq_Model_Pareto_PDF_CDF}) since it clearly illustrates the multi-user diversity gain in caching networks, however, the approach used in this paper can be extended to other distributions in a straightforward manner.

\section{The Proposed Transmission Strategy}\label{Sec_Strategy}

In this section, we first illustrate the network operation with a simple example consisting of $n=4$ sources and destinations, then we will present the general
formulation. First, we explain the concept of \emph{transmission
strategy} in our network setup. Consider the network discussed in Fig. \ref{Fig1}. Each source has $m$ potential destinations, and should first choose one of them at each transmission shot. We define the \emph{best destination vector} $\mathbb{D}$ as the vector representing chosen destinations for different sources.
Suppose we choose the following \emph{best destination vector} $\mathbb{D} = \{D_3, D_1, D_2, D_4\}$.
Such vector indicates that the source $S_1$ decides to transmit to the destination $D_3$, the source $S_2$ to the destination $D_1$, the source $S_3$ to the destination $D_2$, and the source $S_4$ to the destination $D_4$ (obviously we need the destination of $S_i$ to be a member of $\Delta_i$).
For notational convenience, we present the \emph{best destination vector} by $\mathbb{D} = \{D_1^*, D_2^*, D_3^*, D_4^*\}$. So, we have $D_1^*=D_3$, $D_2^*=D_1$, $D_3^*=D_2$, and $D_4^*=D_4$.

Ideally, we want to activate all sources to transmit the messages requested by their desired destinations, as specified in $\mathbb{D}$. However, due to interference, we can only activate a subset of sources,
to ensure successful reception at corresponding destinations. Thus, we should specify which sources should become active by determining the subset of sources belonging to the activated set $\mathbb{S}$. Such assignment determines our \emph{activation strategy}. In the example in
Fig. \ref{Fig1}, we have set the \emph{activation strategy} as $\mathbb{S}=\{S_2,S_4\}$.

The pair $(\mathbb{D},\mathbb{S})$ is defined as the \emph{transmission strategy} of the network. This \emph{transmission strategy} should
be determined based on the network channel state matrix $\mathbf{\Gamma}$, and our objective is to have the largest number of successful transmissions. Once the \emph{transmission strategy} is clear, the sources
send the specified files via a single hop to their intended destinations (see Fig. \ref{Fig1}). In the case of a network with $n$ source-destination pairs, $\mathbb{D}$
is a $1\times n$ vector, and $\mathbb{S}$ can be any subset of $\{S_1,\dots,S_n\}$. Also, throughput will be the number of successful transmissions from each source to its destination, and is a function of the activation strategy $(\mathbb{D},\mathbb{S})$, along with $\mathbf{\Gamma}$.
\begin{algorithm}\label{Alg_Main}
\caption{Determining Transmission Strategy}
\begin{algorithmic}[0]
\\
\\
\State $\mathrm{\mathbf{input}} \gets \mathbf{\Gamma}, \beta, N_0, \{\Delta_i\}_{i=1,\dots,n}$
\\
\\
- - - - - - - - - - - - - - - - - - - - - - - - - - - - - - - - - -
\newline
\Comment{\emph{Determining the best destination vector for each source, and forming $n$ source-destination pairs, namely, $S_i-D_i^*, i=1,\dots,n$.}}
\\- - - - - - - - - - - - - - - - - - - - - - - - - - - - - - - - - -
\State $\mathbb{D} \gets \{\o\} $
\ForAll{$S_i, i=1,\dots,n$}
\State $k^*=\arg{\max_{D_j \in \Delta_i}{\gamma_{i,j}}}$
\State $\gamma_i^*=\gamma_{i,k^*}$
\State $D_i^*=D_{k^*}$
\State $\mathbb{D} \gets \mathbb{D} \cup \{D_i^*\} $
\EndFor
\\
\\
- - - - - - - - - - - - - - - - - - - - - - - - - - - - - - - - - -
\newline
\Comment{\emph{Sorting source-destination pairs (i.e. $S_i-D_i^*$) according to the power of direct link (i.e. $\gamma_i^*$), to arrive at sorted pairs (i.e. $S_{(i)}-D_{(i)}^*$).}}
\\- - - - - - - - - - - - - - - - - - - - - - - - - - - - - - - - - -
\State $\{\gamma_{(1)}^*,\dots,\gamma_{(n)}^*\} \gets \mathrm{sort} \left( \{\gamma_{1}^*,\dots,\gamma_{n}^*\} \right)$
\\
\\
- - - - - - - - - - - - - - - - - - - - - - - - - - - - - - - - - -
\newline
\Comment{\emph{Forming $n$ specific source activation strategies -- i.e. $\mathbb{S}_i, i=1,\dots,n$ -- as the candidate activation sets. In $\mathbb{S}_i$, the first $i$ most powerful pairs are activated. Finally, the one resulting in the largest throughput is selected as the algorithm activation strategy output (i.e. $\mathbb{S}$).} }
\\- - - - - - - - - - - - - - - - - - - - - - - - - - - - - - - - - -

\For{$i=1,\dots,n$}
\State $\mathbb{S}_i \gets \{S_{(n-i+1)},\dots,S_{(n)}\} $
\EndFor
\State $\mathbb{S}_t=\arg\max_{\mathbb{S}_i}{T\left(\mathbb{D},\mathbb{S}_i \right)}$
\State $\mathbb{S} \gets \mathbb{S}_t$
\\
\\
- - - - - - - - - - - - - - - - - - - - - - - - - - - - - - - - - -
\\
\State $\mathrm{\mathbf{output}} \gets \left(\mathbb{D}, \mathbb{S}\right) $
\\
\end{algorithmic}
\end{algorithm}

Next, we discuss the pseudo-code of the algorithm according to which the \emph{transmission strategy} is set, and is presented in Algorithm 1. The input to Algorithm 1 consists of the network channel state matrix $\mathbf{\Gamma}$, $SINR$ threshold $\beta$, AWGN power $N_0$, and the set of potential destinations of each source (according to the data available in source's cache), i.e. $\Delta_i$'s. According to Algorithm 1, first we determine the \emph{best destination vector} $\mathbb{D}=\{D_1^*,\dots,D_n^*\}$. In order to determine $D_i^*$, the source $S_i$ compares the power of its direct links
towards all its potential destinations (i.e. $\gamma_{i,j}$ where $D_j \in \Delta_i$) and selects the destination with the largest direct link power as its corresponding destination. This step assigns destination $D_i^*$ to source $S_i$, which when completed for all sources, forms $n$ source-destination pairs $S_i-D_i^*, i=1,\dots,n$. The direct link power of the pair $S_i-D_i^*$ is called $\gamma_i^*$.

The next step is sorting the power of these direct links (i.e. $\gamma_i^*, i=1,\dots,n $) to derive their order statistics
\begin{equation}\label{Eq_Proposed_Sorted_Powers}
	\gamma_{(1)}^* \leq \gamma_{(2)}^* \leq  \dots \leq \gamma_{(n)}^*.
\end{equation}
Based on the sorted version of direct links (i.e. $\gamma_{(i)}^*, i=1,\dots,n $) we can also sort the source-destination pairs. In other words, we can define
the following ordering between the source-destination pairs:
\begin{equation}\label{Eq_Proposed_Sorted_Pairs}
	S_{(1)}-D_{(1)}^* \leq S_{(2)}-D_{(2)}^* \leq  \dots \leq S_{(n)}-D_{(n)}^*,
\end{equation}
where the power of direct link of the pair $S_{(i)}-D_{(i)}^*$ is $\gamma_{(i)}^*$. In this algorithm, we propose activating the pairs with the most powerful direct links. Thus, we form all the \emph{source activation strategies}
in which the first $i$ most powerful source-destination pairs are activated, where $i$ can range from $1$ to $n$. In other words, we form
\begin{equation}\label{Eq_Proposed_Sample_Activation_Set}
    \mathbb{S}_i = \{S_{(n-i+1)},\dots,S_{(n)}\}, \hspace{10 mm} i=1,\dots,n.
\end{equation}
Subsequently, we select the one which results in the highest throughput
\begin{equation}\label{Eq_Proposed_S_t}
\mathbb{S}_t=\arg\max_{\mathbb{S}_i}{T\left(\mathbb{D},\mathbb{S}_i \right)},
\end{equation}
where $t$ strongest source-destination pairs are activated. Thus, the \emph{activation strategy} will be in the following form
\begin{equation}\label{Eq_Proposed_Optimum_Activation_Set}
    \mathbb{S}=\{S_{(n-t+1)}, \dots, S_{(n)} \}.
\end{equation}
Ideally, we are interested in the case where $t=n$, which is not feasible due to the interference phenomenon. Finally, the \emph{transmission strategy} $\left(\mathbb{D}, \mathbb{S}\right)$ is sent back as the output of the algorithm.

It is interesting to note that we have introduced two stages of sorting in Algorithm 1, in which, we have exploited the most powerful channels. The first one is while determining the \emph{best destination vector}, and the second one is when we determined the \emph{activation strategy}.
These two stages of sorting in Algorithm 1 provides multi-user diversity gain for our scheme which will result in high throughput. In the next section, we will analyze such multi-user diversity gain, and provide a closed-form result for the throughput of the algorithm.

\section{Throughput Analysis}\label{Sec_Throughput}

In this section, we analyze the resulting throughput of Algorithm 1, and investigate the trade-off between the throughput and the cache memory size. The main result of this section is stated in Theorem 1:
\begin{thm}\label{Th_Analysis_Main}
Suppose that each source cache size is of order $m=n^k$ ($0\leq k \leq 1$), and the channel power distribution satisfies Eq. (\ref{Eq_Model_Pareto_PDF_CDF}). Define the resulting throughput of Algorithm 1 as the number of successful receptions at the destinations (denoted by $T$).
Subsequently, we will have\footnote{We say $f(n)=\Omega(g(n))$ if $|f(n)|>k|g(n)|$, for some positive constant $k$ and large-enough $n$.}
\begin{equation}\label{Eq_Analysis_Theorem}
    \mathbb{E}  \left\{ T \right\} = \Omega \left( n^{ (k+1)/(\alpha+1)-\epsilon} \right),
\end{equation}
where $\epsilon>0$ is an arbitrarily small constant, and $\alpha$ characterizes the channel power distribution as indicated in (\ref{Eq_Model_Pareto_PDF_CDF}).
\end{thm}
Since the memory size of each transmitter cache is of order $m=n^k$, Theorem \ref{Th_Analysis_Main} characterizes the
trade-off between the memory size and the network throughput via Eq. (\ref{Eq_Analysis_Theorem}).
\begin{proof}
Define
\begin{equation}\label{Eq_Analysis_Proof_Define_t}
t \triangleq n^{ (k+1)/(\alpha+1)-\epsilon},
\end{equation}
where $\epsilon>0$ is an arbitrarily small constant. In order to prove the theorem, we show that if we activate the first $t$ most powerful source-destination pairs, then, on average, the
order of $\Omega(t)$ of them will be successful. Since Algorithm 1 considers activating most powerful source-destination pairs
of any size, the case where $t$ most powerful pairs are activated will be included in Algorithm 1. This proves that the algorithm achieves the average throughput of order $\Omega(t)$, where $t$ is defined in Eq. (\ref{Eq_Analysis_Proof_Define_t}).

Define
\begin{equation}\label{Eq_Analysis_Proof_Define_r}
r \triangleq n-t+1.
\end{equation}
Then, if $T$ is the number of successful pairs, we will have
\begin{eqnarray}\label{Eq_Analysis_Proof_E_T_Main}
\mathbb{E}\{T\} &=& \sum_{i=r}^{n} { \Pr\left\{\mathrm{SINR}\left(S_{(i)} \rightarrow D_{(i)}^*  \right) \geq \beta\right\} } \\ \nonumber
        &\stackrel{(a)}\geq& t \Pr\left\{\mathrm{SINR}\left(S_{(r)} \rightarrow D_{(r)}^*  \right) \geq \beta\right\} \\ \nonumber
        &\stackrel{(b)}=&t \Pr\left\{ \gamma_{(r)}^* \geq \beta(N_0+I_r) \right\} \\ \nonumber
        &\stackrel{(c)} \geq & t \Pr\left\{ \gamma_{(r)}^* \geq 2\beta\mu t \right\} \Pr \left\{ \beta(N_0+I_r) <2\beta\mu t \right\},
\end{eqnarray}
where (a) is due to the fact that $S_{(r)}-D_{(r)}^*$ is the weakest pair among activated pairs. In (b), $I_r$ represents the interference
imposed on the pair $S_{(r)}-D_{(r)}^*$, and (c) results from independence of $\gamma_{(r)}^*$ and $I_r$ (see \cite{Pooya_WCL} and \cite{Cui}), and $\mu=\mathbb{E}\{\gamma\}$.

In order to proceed, we need the following two lemmas
\begin{lem}\label{Lem_Interference_Ineq}
\begin{equation}\label{Eq_Analysis_Proof_Lemma_Interference_Ineq}
\Pr \left\{ \beta(N_0+I_r) <2\beta\mu t \right\} \geq \frac{1}{2}.
\end{equation}
\end{lem}
\begin{proof}
\begin{eqnarray}\label{Eq_Analysis_Lemma__Interference_Ineq_Proof}
\\ \nonumber
\Pr \left\{ \beta(N_0+I_r) <2\beta\mu t \right\} &=& 1- \Pr \left\{ \beta(N_0+I_r) \geq 2\beta\mu t \right\} \\ \nonumber
&\stackrel{(a)}\geq& 1-\frac{ \mathbb{E}\{\beta(N_0+I_r)\} } {2\beta\mu t} \\ \nonumber
&\stackrel{(b)}=& 1-\frac{ \beta(N_0+\mu (t-1)) } {2\beta\mu t} \\ \nonumber
&\sim& \frac{1}{2} \hspace{5 mm} \mathrm{as} \hspace{5 mm} t\rightarrow \infty,
\end{eqnarray}
where (a) is due to Markov's inequality, and (b) is due to the fact that $I_r$ is the sum of $t-1$ random variables with the mean $\mu$.
\end{proof}
\begin{lem}\label{Lem_Direct_Power_Ineq}
\begin{equation}\label{Eq_Analysis_Proof_Lemma_Direct_Power_Ineq}
\Pr\left\{ \gamma_{(r)}^* \geq 2\beta\mu t \right\} \geq \frac{1}{2}.
\end{equation}
\end{lem}
\begin{proof}
First, we point out the fact that $\gamma_{(r)}^*$, when properly normalized, is a Normal random variable in the following sense:
\begin{equation}\label{Eq_Analysis_Lemma_Direct_Power_Ineq_Falk_1}
    \frac{\gamma_{(r)}^*-a_n}{b_n} \Rightarrow N(0,1),
\end{equation}
where $N(0,1)$ represents the standard Normal random variable, $\Rightarrow$ indicates convergence in distribution, $b_n>0$, and
\begin{equation}\label{Eq_Analysis_Lemma_Direct_Power_Ineq_Falk_2}
    a_n \sim n^{(k+1)/(\alpha+1)} n^{\epsilon/\alpha}.
\end{equation}
The above assertion in (\ref{Eq_Analysis_Lemma_Direct_Power_Ineq_Falk_1}) and (\ref{Eq_Analysis_Lemma_Direct_Power_Ineq_Falk_2}) is rigorously proved in the Appendix. Then, we have
\begin{eqnarray}\label{Eq_Analysis_Lemma_Direct_Power_Ineq_Main}
\\ \nonumber
    \Pr\left\{ \gamma_{(r)}^* \geq 2\beta\mu t \right\} &\stackrel{(a)}=& \Pr\left\{ \gamma_{(r)}^* \geq 2\beta\mu n^{(k+1)/(\alpha+1)-\epsilon} \right\} \\ \nonumber
    &\stackrel{(b)} \geq & \Pr\left\{ \gamma_{(r)}^* \geq  n^{(k+1)/(\alpha+1)}n^{\epsilon/\alpha} \right\} \\ \nonumber
    &\stackrel{(c)}=& \Pr\left\{ \gamma_{(r)}^* \geq  a_n \right\} \\ \nonumber
    &\stackrel{(d)}=&\frac{1}{2},
\end{eqnarray}
where (a) results from (\ref{Eq_Analysis_Proof_Define_t}), (b) is true since we have $n^{\epsilon/\alpha} >  2\beta\mu n^{-\epsilon}$ for large-enough $n$, (c) comes from (\ref{Eq_Analysis_Lemma_Direct_Power_Ineq_Falk_2}),
 and (d) results from (\ref{Eq_Analysis_Lemma_Direct_Power_Ineq_Falk_1}).
\end{proof}
By inserting (\ref{Eq_Analysis_Proof_Lemma_Interference_Ineq}) and (\ref{Eq_Analysis_Proof_Lemma_Direct_Power_Ineq}) into (\ref{Eq_Analysis_Proof_E_T_Main}) we arrive at
\begin{eqnarray}\label{Eq_Analysis_Theorem_Proof_Final}
\mathbb{E}\{T\} \geq \frac{t}{4},
\end{eqnarray}
which states that the average number of successful pairs (i.e. throughput) is lower bounded by $t/4$, concluding the proof.
\end{proof}
\begin{myremark}
It should be noted that, in the proof of Theorem \ref{Th_Analysis_Main} we have implicitly assumed that, at the stage of determining
the \emph{best destination vector}, no destination is chosen by more than one source. One can easily show that, such event is improbable for large networks, and does not
affect our discussions in the proof. For details refer to \cite{Cui_Opportunistic}.
\end{myremark}

\section{Simulation Results}\label{Sec_Simulations}

In order to obtain a deeper understanding of the performance of Algorithm 1, in this section, we provide simulations which
also verify our analytical scaling results presented earlier in Theorem \ref{Th_Analysis_Main}.

Fig. \ref{Fig3} presents the average throughput for the network operating under Algorithm 1, as a function of the number of nodes $n$. The $x$-axis ranges from 30 to 200 nodes, and the results are presented on a log-log scale.
In this figure,  $k=1/2$ indicates that the memory size of each transmitter is of order $m=\sqrt{n}$. The three lines in this figure correspond to different values for $\alpha$, namely, $\alpha=2,3,4$, all with the same memory size. The results are
averaged on $500$ independent network realizations, each operating based on Algorithm 1.

Fig. \ref{Fig3} exhibits some important facts. First, all the three curves (i.e. $\mathbb{E}\{T\}$ vs. $n$) are linear on the log-log scale. This is in agreement
with our result presented in Theorem \ref{Th_Analysis_Main}. The slope of each line should correspond to the scaling exponent of the throughput (as we will clarify later).
Second, the smaller values of $\alpha$ result in higher throughput (for the fixed $k=1/2$). That is due to the fact that by decreasing $\alpha$ we will have a heavier tail for the channel power distribution (see Eq. (\ref{Eq_Model_Pareto_PDF_CDF})).
This will result in more multi-user diversity gain in Algorithm 1, which will result in higher throughput. Such throughput increase is also reflected in the slope of each line in Fig. \ref{Fig3}, which is in agreement with Theorem \ref{Th_Analysis_Main}.

In Fig. \ref{Fig4}, throughput for the same network setting (as in Fig. \ref{Fig3}) is presented for a fixed value of $\alpha=2$, and different values of memory size, namely $m=1,\sqrt{n}$, and $n$, corresponding to $k=0,1/2$, and $1$, respectively. The same linear behaviour of average throughput versus $n$ on the log-log scale is observed in Fig. \ref{Fig4}. Also, this figure shows that for networks with larger transmitter memory size (i.e. larger $k$) we will have
higher throughput. Such throughput increase is also reflected in the slope of each line for different values of $k$, as expected from Theorem \ref{Th_Analysis_Main}. This phenomenon reflects the trade-off between cache size and throughput.

In Fig. \ref{Fig5}, we verify that the slopes of the lines discussed in figures \ref{Fig3} and \ref{Fig4} (i.e. throughput scaling exponents) agree with what we have derived
in Theorem \ref{Th_Analysis_Main}. Define the network throughput scaling exponent as follows
\begin{equation}
e_T \triangleq \frac{\log(\mathbb{E}\{T\})}{\log(n)}.
\end{equation}
In Fig. \ref{Fig5}, $x$-axis represents different values of $k$, $y$-axis represents the network throughput scaling exponent $e_T$, and different lines
correspond to different values of $\alpha$. Dashed lines show what scaling exponent Theorem \ref{Th_Analysis_Main} suggests. Thus, the upper dashed line for $\alpha=2$ is the line $e_T=(k+1)/3$, the middle dashed line for $\alpha=3$ is $e_T=(k+1)/4$, and the lower
dashed line for $\alpha=4$ is $e_T=(k+1)/5$. Solid lines reflect what simulations suggest for the scaling exponent (the slopes of the throughput vs. $n$ plot on log-log scale). As shown in Fig. \ref{Fig5}, there is a good agreement
between the scaling exponents suggested by Theorem \ref{Th_Analysis_Main}, and the simulation results. The increasing nature of $e_T$ as the function of $k$ in Fig. \ref{Fig5}, is again reflecting the memory-throughput trade-off.

\begin{figure}
\begin{center}
\includegraphics[width=0.45\textwidth]{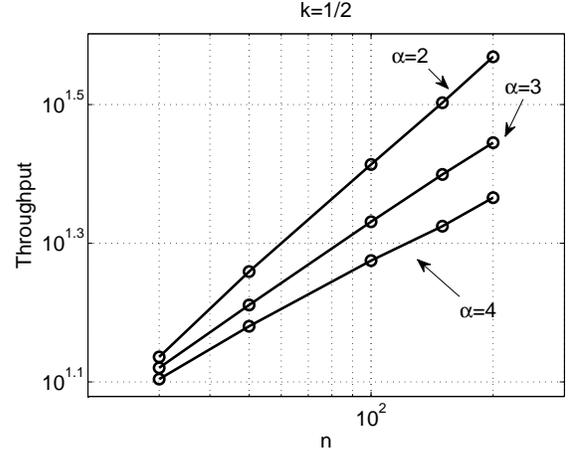}
\end{center}
\caption{Average Throughput vs. $n$ for $k=1/2$ and different values of $\alpha$ on the log-log scale.\label{Fig3}}
\end{figure}

\begin{figure}
\begin{center}
\includegraphics[width=0.45\textwidth]{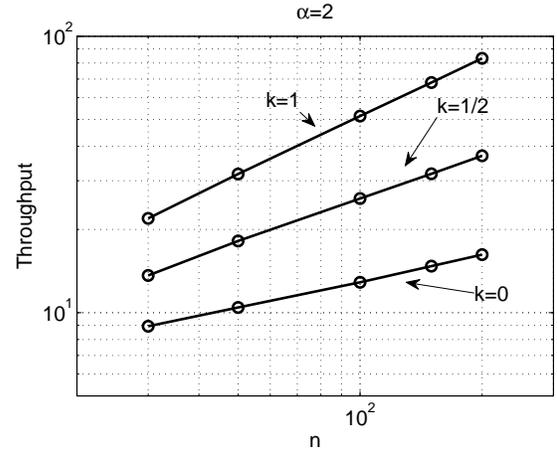}
\end{center}
\caption{Average Throughput vs. $n$ for $\alpha=2$ and different values of $k$ on the log-log scale.\label{Fig4}. The line for $k=0$ refers to the case where there is no caching gain available.}
\end{figure}

\begin{figure}
\begin{center}
\includegraphics[width=0.45\textwidth]{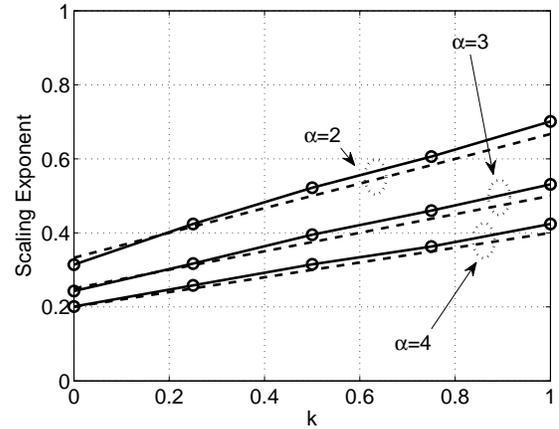}
\end{center}
\caption{Scaling exponent vs. $k$ for different values of $\alpha$. Dashed lines are suggested by Eq. (\ref{Eq_Analysis_Theorem}) in Theorem \ref{Th_Analysis_Main}, and solid lines with circles are resulted from simulations.\label{Fig5}}
\end{figure}

\section{Summary and Future Works}\label{Summary}

In this paper, we have considered the effect of data caching ability, on the throughput of wireless networks. The network we study is under the random connections model,
where fading is the dominant channel effect. Our proposed scheme exploits caching gain in the context of multiuser diversity, by activating wireless channels in better conditions. Based on our results, there
exists a trade-off between cache size and network throughput; the larger memory space available at each transmitter, the higher the network throughput will be. In order to characterize such trade-off, we have proposed a scaling result for network throughput as a function of cache size,
and have verified it by simulations. In this paper, we have considered a heavy-tailed distribution for channel power fading, however, the current approach can be extended to
other distributions in a straightforward manner. In summary, our results show that in cache-enabled wireless networks, we can get substantial network throughput enhancements.

Also, there remain some important issues and future directions for research, one of which is as follows. Since the proposed scheme relies on multi-user diversity gain, each destination will face delay until being served by some source.
Since the network setup is completely symmetrical in the view of all destinations, such delay is fairly imposed to all destinations.
Thus, our scheme is fair over a long time span. However, in networks where such symmetry is violated (e.g. if the channel power gains are not i.i.d. anymore), we should devise
protocols to guarantee fairness. This issue is an interesting future work direction.

\section*{Appendix}\label{Sec_Appendix}
From (\ref{Eq_Proposed_Sorted_Powers}) and (\ref{Eq_Analysis_Proof_Define_r}) it can be verified that $\gamma_{(r)}^*$ is the $t$'th largest value among $\gamma_{i}^*$'s. Thus, we can use the following intermediate order statistics lemma to analyze the distribution of $\gamma_{(r)}^*$
\begin{lem} [Falk, 1989, \cite{Pooya_2013}] \label{Lemma_Falk}
Assume that $X_1,X_2,\dots,X_n$ are i.i.d. random variables with the c.d.f. $F(x)$. Define $X_{(1)},X_{(2)}, \dots, X_{(n)}$ to be the order statistics of $X_1,X_2,\dots,X_n$. If $i \rightarrow \infty$ and $i/n \rightarrow 0$ as $n \rightarrow \infty$, then there exist sequences $a_n$ and $b_n>0$ such that
\begin{equation}\label{Eq_Appendix_Lemma_Falk_1}
	\frac{X_{(n-i+1)}-a_n}{b_n} \Rightarrow N(0,1),
\end{equation}
where $\Rightarrow$ denotes convergence in distribution, and $N(0,1)$ is the Normal distribution with zero mean and unit variance. Furthermore, one choice for $a_n$ and $b_n$ is:
\begin{eqnarray}\label{Eq_Appendix_Lemma_Falk_2}
	a_n=F^{-1}\left(1-\frac{i}{n}\right), \;\;\;	b_n=\frac{\sqrt{i}}{nf(a_n)}.
\end{eqnarray}
\end{lem}
In order to apply Lemma \ref{Lemma_Falk} to the random variable $\gamma_{(r)}^*$, we have to derive $a_n$ for our problem. The first move is to compute
the c.d.f. of the random variable $\gamma_i^*$. Since in the first step of Algorithm 1 (i.e. determining the \emph{best destination vector}), we have chosen the best destination for each source, $\gamma_i^*$ is the maximum of $m$ i.i.d. random variables,
with the distribution given by Eq. (\ref{Eq_Model_Pareto_PDF_CDF}). Therefore, for the c.d.f. of $\gamma_i^*$ we have:
\begin{equation}\label{Eq_Appendix_Max_Distribution}
    F_{\gamma^*}(x)=\left(1-\frac{1}{x^\alpha}\right)^m.
\end{equation}
Thus, in Lemma \ref{Lemma_Falk}, by  setting $i=t$ and applying (\ref{Eq_Appendix_Max_Distribution}) to (\ref{Eq_Appendix_Lemma_Falk_2}), we will have
\begin{eqnarray}\label{Eq_Appendix_our_a_n}
    a_n &=& F^{-1}_{\gamma^*}\left( 1-\frac{t}{n} \right) \\ \nonumber
    &=&  \frac{1}{   \left(1-(1-t/n)^{1/m} \right)^{1/\alpha}    } \\ \nonumber
    &\stackrel{(a)}\sim& \frac{1}{   \left(1-(1-\frac{t}{mn}) \right)^{1/\alpha}    } \\ \nonumber
    &\stackrel{(b)}=& n^{(k+1)/(\alpha+1)} n^{\epsilon/\alpha},
\end{eqnarray}
where (a) follows from the approximation $(1+x)^{\alpha} \sim 1+\alpha x$ for $x \ll 1$, and (b) follows from (\ref{Eq_Analysis_Proof_Define_t}) and our assumption that $m=n^k$. This concludes the proof.

\bibliographystyle{ieeetr}

\end{document}